\documentclass[conference]{IEEEtran}

\usepackage[T1]{fontenc}
\usepackage[utf8]{inputenc}
\usepackage{amsmath,amssymb,amsthm,mathtools}
\usepackage{cite}

\newtheorem{definition}{Definition}
\newtheorem{theorem}{Theorem}
\newtheorem{proposition}{Proposition}
\newtheorem{fact}{Fact}
\newtheorem{corollary}{Corollary}
\newtheorem{remark}{Remark}

\newcommand{\R}{\mathbb{R}}
\newcommand{\Rpp}{\mathbb{R}_{++}}
\newcommand{\C}{\mathbb{C}}
\newcommand{\E}{\mathbb{E}}
\newcommand{\Prb}{\mathbb{P}}

\newcommand{\eml}{\operatorname{eml}}
\newcommand{\depth}{\operatorname{depth}}
\newcommand{\leaves}{\mathcal V_{\mathrm{leaf}}}
\newcommand{\internals}{\mathcal V_{\mathrm{int}}}
\newcommand{\rootnode}{\rho}

\newcommand{\Qfun}{Q}

\begin{document}

\title{EML-AirComp: Layered Over-the-Air Computation from a Single Nomographic Gate}

\author{
\IEEEauthorblockN{Onur G\"unl\"u\textsuperscript{1,2}}
\IEEEauthorblockA{\textsuperscript{1}Lehrstuhl f\"ur Nachrichtentechnik, Technische Universit\"at Dortmund, Germany\\
\textsuperscript{2}Information Theory and Security Laboratory (ITSL), Link\"oping University, Sweden\\[0.5ex]
onur.guenlue@tu-dortmund.de
}
}
\maketitle

\begin{abstract}
Over-the-air computation (AirComp) exploits multiple-access superposition to compute functions of distributed data without
separately decoding all terminal messages. We study a reusable two-input AirComp gate for the exp-minus-log (EML) operation $\eml(u,v)=\exp(u)-\log(v)$, $v>0$. Thus, all internal nodes of a prescribed real-admissible EML tree reuse one gate type, avoiding node-specific nonlinear gate designs. Given an explicit EML tree whose intermediate logarithm arguments remain positive on a given compact domain, we derive additive white Gaussian noise (AWGN) and coherent flat fading implementations under peak-power constraints. We then characterize the number of gate evaluations, the dependency depth, evaluation latency, node-wise feasibility, deterministic error propagation, positivity preservation, and a high-probability AWGN error bound for the complete tree. A four-terminal two-hop example gives explicit positivity and end-to-end error conditions, and a digital interface propagates quantization and gate errors across the tree.
\end{abstract}

\section{Introduction}
Over-the-air computation (AirComp) uses the physical superposition of simultaneous wireless transmissions to compute a function of distributed data directly over the channel.  This idea appears in computation over multiple-access channels (MACs) \cite{NazerGastpar2007}, uncoded transmission over a Gaussian sensor network \cite{Gastpar2008}, compute-and-forward \cite{NazerGastpar2011}, analog function computation over wireless MACs \cite{GoldenbaumBocheStanczak2013,GoldenbaumStanczak2013}, and nomographic function computation in clustered Gaussian sensor networks \cite{GoldenbaumBocheStanczak2015}.  It has also become a central primitive in wireless learning, including federated edge learning and collaborative inference \cite{ZhuDuGunduzHuang2021,SahinYang2023,PerezNeira2025,YilmazHasirciogluQiaoGunduz2024}.

A standard AirComp-compatible class is the class of nomographic functions. Such a function can be written as
\begin{align}
    f(s_1,\ldots,s_K) = \psi\left(\sum_{k=1}^{K} \phi_k(s_k)\right),\label{eq:nomographic}
\end{align}
where terminal $k\in[1:K]$ applies the local pre-processing function $\phi_k(\cdot)$, the MAC provides the sum, and the receiver applies the post-processing function $\psi(\cdot)$. This model directly covers sums, weighted averages, and other functions with a post-processed-sum representation. For multivariate inputs, however, continuous functions admitting a representation of the form \eqref{eq:nomographic} form a nowhere-dense subset of the continuous functions under the uniform norm~\cite{Buck1982}. A broader approach is to compute a finite sequence of nomographic functions. In wireless computation, this finite-succession viewpoint appears in clustered Gaussian sensor networks~\cite{GoldenbaumBocheStanczak2015}. This motivates the question studied here: \textit{can nonlinear scalar operations around AirComp sums be evaluated by repeatedly using the same AirComp gate?}

We consider the exp-minus-log (EML) operation, defined as $\eml(u,v)=\exp(u)-\log(v)$ for $v>0$. A recent result gives constructive EML representations of the operations in a specified scientific-calculator basis using the constant $1$ and repeated applications of the EML operation \cite{Odrzywolek2026}. Some of those representations require complex intermediate values, whereas
the AirComp model considered here is real-valued. We therefore
use only EML identities whose complete trees are real-valued and whose right-child values remain strictly positive on the prescribed operating domain. To this end, we map explicit real-admissible EML trees to layered AirComp evaluation schedules. Each internal node of the tree is evaluated by the same two-input EML-AirComp gate. For any fixed tree on a compact operating domain, the construction gives the number of gate evaluations, the dependency depth, per-node transmit-amplitude bounds, deterministic error propagation, and positivity margins for all logarithm inputs. We also specialize the same gate to additive white Gaussian noise (AWGN) and coherent flat fading MACs under peak-power constraints. Here, reuse refers to the functional gate type: every internal node uses the same nonlinear pre-processing maps $\exp(\cdot)$ and $-\log(\cdot)$ and the same receiver normalization structure, while the operating intervals, scaling factors, and channel coefficients may be node dependent.

The main contributions are as follows. First, we give peak-power-feasible AWGN and coherent flat-fading gate implementations with explicit scaling, mean-squared error (MSE), Gaussian-tail, and Rayleigh channel-inversion formulas. Second, we show that a fixed real-admissible EML tree $T$ requires one gate evaluation per internal node, has critical-path length $\depth(T)$, and admits explicit resource-latency bounds. Third, we derive deterministic path-product and high-probability AWGN bounds with sufficient positivity conditions, and illustrate them through a three-gate hierarchical example and a generic digital-error interface.

\section{Single EML Gate over AWGN and Fading Channels} \label{sec:single_gate}
This section provides AWGN and coherent flat-fading realizations of one EML-AirComp gate under peak-power constraints.

\subsection{Nomographic Form and Bounded Operating Intervals}
Let $u\in I=[a,b]$ and $v\in J=[c,d]$, where $c>0$. As the right AirComp input contributes a scaled version of $-\log(v)$, we define
\begin{align}
    B_J=\max\{|\log c|,|\log d|\},
\end{align}
which is used to ensure that the transmitted signal satisfies the peak-power constraint. If this bound is zero, then $\log(v)=0$ over the considered interval, so the corresponding transmitted waveform is zero. In that case, the associated peak-power constraint does not restrict the scaling parameter and we use the convention $\sqrt{P}/0=\infty$ for this special case. The following fact provides the exact nomographic representation of the EML gate and the channel-input bounds used in the AWGN and fading implementations below.

\begin{fact}\label{fact:eml_gate}
For $(u,v)\in I\times J$, the function
\begin{align}
    G(u,v)=\exp(u)-\log(v)
\end{align}
is nomographic with
\begin{align}
    G(u,v)=\psi\big(\phi_{\mathrm L}(u)+\phi_{\mathrm R}(v)\big),
\end{align}
where
\begin{align}
    \phi_{\mathrm L}(u)=\exp(u),\quad
    \phi_{\mathrm R}(v)=-\log(v),\quad
    \psi(y)=y .
\end{align}
With scaling $A>0$, the two real channel inputs used by the EML-AirComp gate are
\begin{align}
    X_{\mathrm L}=A\exp(u),\qquad X_{\mathrm R}=-A\log(v) \label{eq:XRandXLforsclaingA}
\end{align}
which satisfy
\begin{align}
    |X_{\mathrm L}|\le A \exp(b),\qquad
    |X_{\mathrm R}|\le A B_J.
\end{align}
\end{fact}

These bounds are used below to choose the scaling $A$ so that the corresponding peak-power constraints are satisfied.

\subsection{AWGN EML-AirComp Gate}
We now specialize the EML gate to a real two-user AWGN MAC. The scaling parameter is chosen subject to the peak-power constraints.

Consider the real two-user AWGN multiple-access gate
\begin{align}
    Y=X_{\mathrm L}+X_{\mathrm R}+Z,
    \qquad Z\sim\mathcal N(0,\sigma^2),
    \label{eq:awgn_gate_model}
\end{align}
with peak-power constraints $|X_{\mathrm L}|^2\le P_{\mathrm L}$ and $|X_{\mathrm R}|^2\le P_{\mathrm R}$.  For scaling $A>0$, set \eqref{eq:XRandXLforsclaingA} and form the estimate 
\begin{align}
    \widehat G=Y/A.
\end{align}

The next result gives the exact AWGN error law and the MSE-minimizing scaling under peak-power constraints, where $Q(x)$ denotes the Gaussian tail function. Throughout the AWGN and fading models below, assume $P_{\mathrm L},P_{\mathrm R}>0$.

\begin{fact}\label{fact:awgn_gate}
Consider the AWGN channel. If we have
\begin{align}
    0<A\le A_{\max}^{\mathrm{AWGN}}
    \triangleq
    \min\left\{\frac{\sqrt{P_{\mathrm L}}}{\exp(b)},
    \frac{\sqrt{P_{\mathrm R}}}{B_J}\right\},
    \label{eq:awgn_A_constraint}
\end{align}
then \eqref{eq:XRandXLforsclaingA} satisfies the peak-power constraints for all $(u,v)\in I\times J$. Moreover, we have
\begin{align}
    \widehat G = \frac{Y}{A} = G(u,v)+\xi,\qquad \xi=\frac{Z}{A}\sim \mathcal N\left(0,\frac{\sigma^2}{A^2}\right)
\end{align}
so that we obtain
\begin{align}
    \E\big[(\widehat G-G)^2\big]=\frac{\sigma^2}{A^2}
    \label{eq:awgn_mse}
\end{align}
and, for every $\epsilon>0$, we have
\begin{align}
    \Prb\big[|\widehat G-G|>\epsilon\big]
    =2\Qfun\!\left(\frac{A\epsilon}{\sigma}\right).
    \label{eq:awgn_tail}
\end{align}
Among all feasible $A$, the minimum local MSE is achieved by $A=A_{\max}^{\mathrm{AWGN}}$ and is equal to
\begin{align}
    \mathrm{MSE}_{\min}^{\mathrm{AWGN}}
    =\frac{\sigma^2}{(A_{\max}^{\mathrm{AWGN}})^2}.
\end{align}
\end{fact}

\begin{IEEEproof}
The peak-power constraints follow from Fact~\ref{fact:eml_gate}.  Substituting \eqref{eq:XRandXLforsclaingA} into \eqref{eq:awgn_gate_model} gives 
\begin{align}
    Y=A(\exp(u)-\log(v))+Z,
\end{align}
which proves the Gaussian error law, tail probability, and local MSE. Since $\sigma^2/A^2$ is decreasing in $A>0$, the MSE-minimizing feasible scaling is the largest feasible scaling.
\end{IEEEproof}

\subsection{Coherent Flat-Fading EML-AirComp Gate}
We next consider a coherent flat-fading MAC. Each transmitter applies channel inversion so that the two pre-processed EML inputs add coherently at the receiver.

Consider the complex flat-fading MAC
\begin{align}
    Y=h_{\mathrm L}X_{\mathrm L}+h_{\mathrm R}X_{\mathrm R}+Z,
    \label{eq:fading_gate_model}
\end{align}
where $h_{\mathrm L},h_{\mathrm R}\in\C\setminus\{0\}$ are independent fading coefficients known at the corresponding transmitters, and $Z$ is independent of $(h_{\mathrm L},h_{\mathrm R})$. Let $\sigma_R^2=\E[(\Re\{Z\})^2]$. Use
\begin{align}
    X_{\mathrm L}=A\frac{h_{\mathrm L}^*}{|h_{\mathrm L}|^2}\exp(u),
    \qquad
    X_{\mathrm R}=-A\frac{h_{\mathrm R}^*}{|h_{\mathrm R}|^2}\log(v),
    \label{eq:fading_eml_tx}
\end{align}
and form the estimate 
\begin{align}
    \widehat G =G(u,v)+\xi,\qquad \xi=\frac{\Re\{Z\}}{A}.
\end{align}
If $\Re\{Z\}\sim\mathcal N(0,\sigma_R^2)$, then we have
\begin{align}
    \xi\sim\mathcal N\left(0,\frac{\sigma_R^2}{A^2}\right).
\end{align}

The next result analyzes the coherent flat fading gate under channel inversion. It gives the conditional MSE and tail probability for fixed fading coefficients, and the feasibility probability that a fixed scaling satisfies the peak-power constraints under independent Rayleigh fading.

\begin{fact}\label{fact:fading_gate}
Consider the coherent flat-fading MAC. If we have
\begin{align}
    0<A\le A_{\max}^{\mathrm{fad}}
    \triangleq
    \min\left\{\frac{\sqrt{P_{\mathrm L}}|h_{\mathrm L}|}{\exp(b)},
    \frac{\sqrt{P_{\mathrm R}}|h_{\mathrm R}|}{B_J}\right\},
    \label{eq:fading_A_constraint}
\end{align}
then \eqref{eq:fading_eml_tx} satisfies the peak-power constraints for all $(u,v)\in I\times J$.  The conditional MSE is
\begin{align}
    \E\big[(\widehat G-G)^2\,|\,h_{\mathrm L},h_{\mathrm R}\big]
    =\frac{\sigma_R^2}{A^2}.
    \label{eq:fading_mse}
\end{align}
If $\Re\{Z\}\sim\mathcal N(0,\sigma_R^2)$, then we have
\begin{align}
    \Prb\big[|\widehat G-G|>\epsilon\,|\,h_{\mathrm L},h_{\mathrm R}\big]
    =2\Qfun\!\left(\frac{A\epsilon}{\sigma_R}\right).\label{eq:fading_tailprobability}
\end{align}
Among all feasible $A$ for the realized channel, the conditional MSE is minimized by $A=A_{\max}^{\mathrm{fad}}$. Since $|h_{\mathrm L}|^2$ and $|h_{\mathrm R}|^2$ are independent exponential random variables with some means $\Omega_{\mathrm L}$ and $\Omega_{\mathrm R}$, respectively, then a fixed scaling $A$ is feasible with probability
\begin{align}
    p_{\mathrm{feas}}(A)
    =\exp\!\left(
    -\frac{A^2 \exp(2b)}{P_{\mathrm L}\Omega_{\mathrm L}}
    -\frac{A^2 B_J^2}{P_{\mathrm R}\Omega_{\mathrm R}}
    \right).
    \label{eq:rayleigh_feas}
\end{align}
\end{fact}

\begin{IEEEproof}
Substituting \eqref{eq:fading_eml_tx} into \eqref{eq:fading_gate_model} gives
\begin{align}
    \widehat G=\frac{\Re\{Y\}}{A}=G(u,v)+\frac{\Re\{Z\}}{A},
\end{align}
which yields \eqref{eq:fading_mse} and \eqref{eq:fading_tailprobability}. Moreover, using $u\le b$ and $|\log(v)|\le B_J$, we obtain
\begin{align}
    |X_{\mathrm L}|^2&\le\frac{A^2\exp(2b)}{|h_{\mathrm L}|^2}\le P_{\mathrm L},\quad
    |X_{\mathrm R}|^2\le\frac{A^2B_J^2}{|h_{\mathrm R}|^2}\le P_{\mathrm R},
\end{align}
where the last inequalities follow from \eqref{eq:fading_A_constraint}. Since $\sigma_R^2/A^2$ decreases with $A>0$, the conditional MSE is minimized by $A=A_{\max}^{\mathrm{fad}}$. Moreover, for a fixed $A$, feasibility for all $(u,v)\in I\times J$ is equivalent to
\begin{align}
    |h_{\mathrm L}|^2\ge\frac{A^2\exp(2b)}{P_{\mathrm L}},\qquad
    |h_{\mathrm R}|^2\ge\frac{A^2B_J^2}{P_{\mathrm R}}.
\end{align}
Since the channel power gains are independent exponential random variables, their tail probabilities give \eqref{eq:rayleigh_feas}.
\end{IEEEproof}

\section{Explicit EML Trees and Layered AirComp Schedules}\label{sec:trees}
Having established the implementation of one EML-AirComp gate, we now extend it to a fixed EML tree, where each internal node requires one gate evaluation and the longest leaf-to-root path determines the dependency depth.

The following definition fixes the tree notation used in all later bounds.

\begin{definition}
An EML computation tree $T$ over inputs $s=(s_1,\ldots,s_K)$ is a finite binary computation tree whose internal nodes are EML gates. The leaves provide the starting values of the computation, and the root node gives the final output. Let $\leaves(T)$ and $\internals(T)$ denote the sets of leaves and internal nodes, respectively, and let $\rootnode$ denote the root. Each leaf is either one of the input variables $s_k$ or a copy of the constant $1$. Thus, for a leaf $i$, the associated value is either
$F_i(s)=s_k$ for some $k\in[1:K]$, or $F_i(s)=1$. Each internal node $q\in\internals(T)$ has exactly two children: a left child $\ell(q)$ and a right child $r(q)$. The value computed at node $q$ is defined recursively as
\begin{align}
    F_q(s)=G\big(F_{\ell(q)}(s),F_{r(q)}(s)\big).
\end{align}
The function represented by the whole tree is the value at the root $F_T(s)=F_{\rootnode}(s)$. The number of EML gates in the tree is $W(T)$, and the depth $\depth(T)$ is the largest number of EML gates on any path from a leaf to the root.\hfill $\lozenge$
\end{definition}

The next definition gives the real-admissibility condition, needed as the second argument of $G$ enters a logarithm.

\begin{definition}
Let $D\subset\R^K$.  An EML tree $T$ is real-admissible on $D$ if all node functions are recursively well-defined and real-valued on $D$, and if we have
\begin{align}
    F_{r(q)}(s)>0, \qquad \forall s\in D, \quad \forall q\in\internals(T).
\end{align}\hfill $\lozenge$
\end{definition}

For each internal node $q$, we need bounds on the values that can appear at its two inputs. As $s$ ranges over the operating domain $D$, the left child of node $q$ can take values in $I_q=F_{\ell(q)}(D)$ and the right child can take values in $J_q=F_{r(q)}(D)$. If $D$ is compact and $T$ is real-admissible on $D$, then all node functions are continuous. Thus, the sets $I_q$ and $J_q$ are compact. Moreover, real-admissibility guarantees that the right input of every EML gate is strictly positive on $D$. In the analysis, we use intervals that contain the following value ranges:
\begin{align}
    I_q\subseteq[a_q,b_q],\qquad J_q\subseteq[c_q,d_q]\subset\Rpp. \label{eq:node_intervals}
\end{align}
Thus, $B_q=\max\{|\log c_q|,|\log d_q|\}$ is an upper bound on $|\log v|$ at node $q$.

An evaluation order is said to be \textit{admissible} if each internal node is evaluated only after both child values are available. Moreover, at each internal node, the two child values are assumed to be available at the transmitters feeding one EML-AirComp gate. Routing, storage, forwarding, synchronization, and waveform generation between gate evaluations are not modeled. Hence, $W(T)$ counts only EML-AirComp gate evaluations.

The next result counts the required EML-AirComp gate evaluations, identifies the tree depth as the critical-path latency, and gives node-wise peak-power feasibility conditions.

\begin{proposition}\label{prop:compiler}
Let $D\subset\R^K$ be compact and let $T$ be real-admissible on $D$. Consider any evaluation order in which an internal node is evaluated only after the values of its two children have already been computed. Then, the function $F_T$ is evaluated by one EML-AirComp gate evaluation for each internal node of $T$. Hence, the total number of EML-AirComp gate evaluations is
$W(T)=|\internals(T)|$. A critical-path lower bound on the
number of sequential gate-evaluation intervals is $\depth(T)$, achieved when all nodes in each dependency layer can be evaluated in parallel. 

For each internal node $q$, let the possible left- and right-child values be enclosed by the intervals in \eqref{eq:node_intervals}. If node $q$ is implemented over the real AWGN gate with scaling $A_q$ and peak-power limits $P_{q,\mathrm L}$ and $P_{q,\mathrm R}$, then it is sufficient to choose
\begin{align}
0<A_q\le\min\left\{\frac{\sqrt{P_{q,\mathrm L}}}{\exp(b_q)},\frac{\sqrt{P_{q,\mathrm R}}}{B_q}\right\}.
\label{eq:tree_awgn_A}
\end{align}
Similarly, if node $q$ is implemented over the coherent flat-fading gate with channel coefficients $h_{q,\mathrm L}$ and $h_{q,\mathrm R}$, then it is sufficient to choose
\begin{align}
0<A_q\le\min\left\{\frac{\sqrt{P_{q,\mathrm L}}|h_{q,\mathrm L}|}{\exp(b_q)},\frac{\sqrt{P_{q,\mathrm R}}|h_{q,\mathrm R}|}{B_q}\right\}.
\label{eq:tree_fading_A}
\end{align}
These guarantee the peak-power constraints.
\end{proposition}

\begin{IEEEproof}
Each internal node computes exactly one value of the form $G(u,v)$.  Evaluating the nodes in any topological order, therefore, evaluates the root after all internal nodes have been evaluated. The number of gate invocations is $|\internals(T)|=W(T)$, and the longest dependency chain is the maximum number of internal nodes on a root-to-leaf path.  The node-wise AWGN and fading constraints follow from Facts~\ref{fact:awgn_gate} and \ref{fact:fading_gate} applied to the enclosing intervals of node $q$.
\end{IEEEproof}

\begin{remark}\label{rem:parallelism_latency}
Assume that each EML-AirComp gate evaluation takes one interval and that at most $P\ge1$ gates can be evaluated in parallel. Let $W_\ell$ be the number of internal nodes whose longest leaf-to-node path contains $\ell$ gates, and let $D_P^\star(T)$ be the minimum number of intervals required by an admissible schedule. Then, we have
\begin{align}
    \max\left\{\depth(T),\left\lceil\frac{W(T)}{P}\right\rceil\right\}\le D_P^\star(T)\le\sum_{\ell=1}^{\depth(T)}\left\lceil\frac{W_\ell}{P}\right\rceil
\label{eq:scheduling_bounds}
\end{align}
where the lower bounds follow from the longest dependency path and the total number of gates, while evaluating one layer at a time gives the upper bound. Thus, $D_1^\star(T)=W(T)$, and $D_P^\star(T)=\depth(T)$ if $P\ge\max_\ell W_\ell$.
\end{remark}

For a noisy sequential evaluation, the same node-wise conditions apply on any bounded-error event for which the perturbed operands remain in the prescribed intervals. Outside this event, no downstream peak-power guarantee is made, and the evaluation may be declared to be in outage, as discussed in the next section.

\section{Tree-Level Error Propagation, Positivity, and Reliability}
\label{sec:tree_error}
This section derives deterministic and high-probability AWGN tree-error bounds while ensuring positive logarithm arguments.

Let $\widehat F_q$ denote the noisy computed value at node $q$. For each internal node $q$, define the analysis rectangle
\begin{align}
    \mathcal R_q=[a_q,b_q]\times[c_q,d_q]\subset\R\times\Rpp,\label{eq:rectangles}
\end{align}
over which the local sensitivity bounds are evaluated. Define
\begin{align}
    L_q^{\mathrm L}=\exp(b_q),\qquad L_q^{\mathrm R}=\frac{1}{c_q}
\end{align}
where $L_q^{\mathrm L}$ bounds how sensitive the output of node $q$ is to an error in its left input, and $L_q^{\mathrm R}$ in its right input. For a connection $p\to q$, where $p$ is a child of $q$, set the sensitivity factor as
\begin{align}
    L_{p\to q}=\begin{cases}
                L_q^{\mathrm L}, & p=\ell(q),\\
                L_q^{\mathrm R}, & p=r(q).
                \end{cases}
\end{align}
Moreover, for a node $p$, let $P(p\to\rootnode)$ denote the unique path from $p$ to the root, and define
\begin{align}
    M(p\to\rootnode)=\prod_{e\in P(p\to\rootnode)}L_e,
\end{align}
with the empty product equal to $1$, such that it is the product of the sensitivity factors along the path from $p$ to the root, and therefore bounds how much an error introduced at node $p$ can be amplified before it reaches the final output. 

The next theorem gives a deterministic bound on how local errors propagate through the tree. It also provides checkable range conditions ensuring that every perturbed logarithm argument remains positive. To initialize a common error recursion over all tree nodes, set the error budget of each leaf occurrence $i$ to its given error bound, i.e., $\delta_i=\varepsilon_i$, and define the error budget of every internal node $q$ recursively as
\begin{align}
    \delta_q=L_q^{\mathrm L}\delta_{\ell(q)}+L_q^{\mathrm R}\delta_{r(q)}+\eta_q.\label{eq:recursive_error_budget}
\end{align}

\begin{theorem}\label{thm:path_bound}
Let $D\subset\R^K$ be compact, and let $T$ be real-admissible on $D$. Suppose that every leaf occurrence
$i\in\leaves(T)$ satisfies 
\begin{align}
    |\widehat F_i(s)-F_i(s)|\leq \varepsilon_i, \qquad s\in D,    \label{eq:leaf_error_bound}
\end{align}
and that the noisy evaluation is performed in an admissible order, with every internal node $q\in\internals(T)$ satisfying
\begin{align}
    \widehat F_q(s)  = G\big(\widehat F_{\ell(q)}(s),\widehat F_{r(q)}(s) \big) + \xi_q(s), \qquad
    |\xi_q(s)|\leq\eta_q. \label{eq:internal_error_model}
\end{align}
For every internal node $q$, suppose 
\begin{align}
    F_{\ell(q)}(D)    &\subseteq [a_q+\delta_{\ell(q)},\, b_q-\delta_{\ell(q)}], \label{eq:buffered_left_range}\\
    F_{r(q)}(D)  &\subseteq [c_q+\delta_{r(q)},\,
     d_q-\delta_{r(q)}]. \label{eq:buffered_right_range}
\end{align}
Then, every noisy gate evaluation is well defined, the exact and computed child-value pairs remain in $\mathcal R_q$, and
\begin{align}
    |\widehat F_q(s)-F_q(s)|    \leq\delta_q,
    \qquad    q\in\internals(T),\quad s\in D.\label{eq:node_error_budget}
\end{align}
In particular, we have
\begin{align}
    |\widehat F_T(s)-F_T(s)|&\leq \sum_{i\in\leaves(T)}
    M(i\to\rootnode)\varepsilon_i\nonumber\\
    &\quad+\sum_{q\in\internals(T)}M(q\to\rootnode)\eta_q.    \label{eq:path_bound}
\end{align}
Moreover, every computed logarithm argument satisfies
\begin{align}
    \widehat F_{r(q)}(s)\in[c_q,d_q]\subset\Rpp,
    \qquad   q\in\internals(T),\quad s\in D.\label{eq:computed_positivity}
\end{align}
\end{theorem}

\begin{IEEEproof}
Fix $s\in D$ and consider any admissible evaluation order. For every leaf occurrence $i$, \eqref{eq:leaf_error_bound}, and the definition $\delta_i=\varepsilon_i$ give
\begin{align}
    |\widehat F_i(s)-F_i(s)|\leq\delta_i.\label{eq:leaf_induction}
\end{align}

Consider an internal node $q$ and assume that the bound
\eqref{eq:node_error_budget} has already been established for its two children. From \eqref{eq:buffered_left_range}, we have $F_{\ell(q)}(s)    \in [a_q+\delta_{\ell(q)},\, b_q-\delta_{\ell(q)}]$. Using $|\widehat F_{\ell(q)}(s)-F_{\ell(q)}(s)| \leq\delta_{\ell(q)}$, this implies
\begin{align}
    \widehat F_{\ell(q)}(s)\in[a_q,b_q].
    \label{eq:computed_left_containment}
\end{align}
Similarly, \eqref{eq:buffered_right_range} and
$|\widehat F_{r(q)}(s)-F_{r(q)}(s)|\leq\delta_{r(q)}$ give
\begin{align}
    \widehat F_{r(q)}(s)\in[c_q,d_q]\subset\Rpp.
    \label{eq:computed_right_containment}
\end{align}
Thus, both the exact and computed child-value pairs belong to the rectangle $\mathcal R_q    =   [a_q,b_q]\times[c_q,d_q]$, and the logarithm in the noisy evaluation of node $q$ is well defined.

Consider arbitrary $(u,v),(u',v')\in\mathcal R_q$. Since $\mathcal R_q$ is convex, we have
\begin{align}
    |G(u,v)-G(u',v')| \leq \exp(b_q)|u-u'|  +  \frac{1}{c_q}|v-v'|.   \label{eq:local_lipschitz_bound}
\end{align}
Applying \eqref{eq:local_lipschitz_bound} to the exact and computed child-value pairs at node $q$, and using
\eqref{eq:internal_error_model}, gives
\begin{align}
    |\widehat F_q(s)-F_q(s)|   &\leq  L_q^{\mathrm L}
    |\widehat F_{\ell(q)}(s)-F_{\ell(q)}(s)| \nonumber\\
    &\quad+   L_q^{\mathrm R}  |\widehat F_{r(q)}(s)-F_{r(q)}(s)|   +  \eta_q   \nonumber\\
    &\leq    L_q^{\mathrm L}\delta_{\ell(q)}    +L_q^{\mathrm R}\delta_{r(q)}    +  \eta_q=\delta_q.
    \label{eq:error_recursion}
\end{align}
This proves \eqref{eq:node_error_budget} by induction over the admissible evaluation order. Equation \eqref{eq:computed_right_containment} simultaneously proves \eqref{eq:computed_positivity}.

Recursively expanding \eqref{eq:recursive_error_budget} from the root toward the leaves shows that every leaf error $\varepsilon_i$ is multiplied by the product of the sensitivity factors on the unique path from leaf occurrence $i$ to the root. Similarly, every local gate error $\eta_q$ is multiplied by the product of the sensitivity factors on the path from node $q$ to the root. By the definition of $M(\cdot\!\to\!\rootnode)$, giving \eqref{eq:path_bound}.
\end{IEEEproof}

We next specialize the deterministic tree bound to independent AWGN-induced errors at the EML-AirComp nodes, and give a lower bound on the probability that the final root error satisfies the deterministic path-product bound.

\begin{corollary}\label{cor:hp_awgn_tree}
Fix $s\in D$ and assume the setting of Theorem~\ref{thm:path_bound}, except that the internal-node errors are independent AWGN-induced errors such that
\begin{align}
    \xi_q\sim\mathcal N\!\left(0,\frac{\sigma_q^2}{A_q^2}\right),
    \qquad q\in\internals(T).
\end{align}
Choose local tolerances $\eta_q>0$. Assume that the buffered range conditions
\eqref{eq:buffered_left_range} and
\eqref{eq:buffered_right_range} hold for the selected local tolerances $\eta_q$. Then, for this fixed input $s$, the bound
\begin{align}
    |\widehat F_T(s)-F_T(s)|
    &\le \sum_{i\in\leaves(T)}M(i\to\rootnode)\varepsilon_i\nonumber\\
    &\quad+\sum_{q\in\internals(T)}M(q\to\rootnode)\eta_q
    \label{eq:hp_awgn_tree_bound}
\end{align}
holds with probability at least
\begin{align}
    \prod_{q\in\internals(T)}
    \left[
    1-2\Qfun\!\left(\frac{A_q\eta_q}{\sigma_q}\right)
    \right].
    \label{eq:hp_awgn_tree}
\end{align}
On the same event, every logarithm argument remains in its positive interval, and no interval-outage event occurs.
\end{corollary}

\begin{IEEEproof}
For every internal node $q$, Fact~\ref{fact:awgn_gate} gives
\begin{align}
    \Prb[|\xi_q|\leq\eta_q] = 1-2\Qfun\!\left(\frac{A_q\eta_q}{\sigma_q}
    \right).
\end{align}
Since the node errors are independent, the probability that all events $\{|\xi_q|\leq\eta_q\}$ occur is the product in \eqref{eq:hp_awgn_tree}. On this event, the local error bounds required by Theorem~\ref{thm:path_bound} hold simultaneously. The buffered range conditions then ensure that all gate
evaluations are well defined and that all computed child-value pairs remain in their assigned rectangles. 
\end{IEEEproof}

\subsection{Finite-Alphabet and Digital Gate Errors}
The preceding deterministic bound also applies when the leaf values are quantized and the individual EML gates are evaluated by finite-alphabet methods. Let $\mathcal Q:D\to D_{\mathcal Q}\subset D$ be a quantizer with
finite image $D_{\mathcal Q}$, and write $s^{\mathcal Q}=\mathcal Q(s)$. For each coordinate-leaf occurrence $i$, assume 
\begin{align}
    |s_{k(i)}-s^{\mathcal Q}_{k(i)}|\le  \Delta_{k(i)}.
\end{align}
Suppose also that a finite-alphabet or digital AirComp method is used for every EML node $q$, and let $\mathcal E_q$ be an event on which its local computation error satisfies $|\xi_q|\le\eta_q$. Let also $\widehat F_T(s^{\mathcal Q})$ denote the output obtained by evaluating the perturbed digital EML tree from the quantized leaf values.

The following corollary combines the quantization errors at the leaves with method-dependent errors at the digital EML gates in one end-to-end path-product bound.

\begin{corollary}\label{cor:digital_interface}
Let $D\subset\R^K$ be compact, let $T$ be real-admissible on $D$, and fix $s\in D$. Suppose that the buffered range conditions of Theorem~\ref{thm:path_bound} hold with
$\varepsilon_i=\bar\Delta_i$ and with the selected gate-error tolerances $\eta_q$. Define
\begin{align}
    \bar\Delta_i = \begin{cases}
    \Delta_{k(i)},
    & i \text{ is a coordinate-leaf occurrence},\\
    0,
    & i \text{ is a constant-}1\text{ leaf}.
    \end{cases}
\end{align}
Then, on $\bigcap_{q\in\internals(T)}\mathcal E_q$, we have
\begin{align}
    |\widehat F_T(s^{\mathcal Q})-F_T(s)|    &\le\sum_{i\in\leaves(T)}M(i\to\rootnode)\bar\Delta_i
    \nonumber\\
&\quad+\sum_{q\in\internals(T)}M(q\to\rootnode)\eta_q.
    \label{eq:digital_combined_bound}
\end{align}
If $\Prb[\mathcal E_q^c]\!\le\! p_q$, \eqref{eq:digital_combined_bound} holds with probability at least
\begin{align}
    1\!-\!\sum_{q\in\internals(T)}p_q.
\end{align}
\end{corollary}

\begin{IEEEproof}
Apply Theorem~\ref{thm:path_bound} on $\bigcap_{q\in\internals(T)}\mathcal E_q$ with $\varepsilon_i=\bar\Delta_i$, and use the union bound.
\end{IEEEproof}

Note that Corollary~\ref{cor:digital_interface} does not provide a finite-alphabet coding method. If applied to the selected finite input alphabets, a method such as ChannelComp or SumComp \cite{RazavikiaDaSilvaFischione2024,
RazavikiaDaSilvaFischione2025} can separately provide the method-dependent local error tolerances $\eta_q$ or failure probabilities $p_q$. The corresponding quantization, coding, modulation, and decoding analyses remain separate from the EML-tree error propagation in \eqref{eq:digital_combined_bound}.

\section{Hierarchical EML-AirComp Example} \label{sec:hiererchicalexample}
Consider four terminals, two relays, and one fusion center. Terminals $1$ and $2$ transmit to relay $1$, terminals $3$ and $4$ to relay $2$, and both relays transmit to the fusion center. The first-hop MACs use noninterfering resources and may operate simultaneously. At each EML gate, the two operands are held by different transmitters and combined over a MAC.

Let $s_1\in[a_1,b_1]$, $s_2\in[c_2,d_2]\subset\Rpp$, $s_3\in[a_3,b_3]$, and $s_4\in[c_4,d_4]\subset\Rpp$. Define
\begin{align}
    U_1\!=\!G(s_1,s_2),\;\; U_2\!=\!G(s_3,s_4),\;\; F_T(s)\!=\!G(U_1,U_2).\label{eq:hier_tree}
\end{align}
Since $G(u,v)$ is increasing in $u$ and decreasing in $v$, the exact first-level output intervals are determined by
\begin{align}
    &\underline u_1=\exp(a_1)-\log d_2,\quad \overline u_1=\exp(b_1)-\log c_2,\nonumber\\
    &\underline u_2=\exp(a_3)-\log d_4,\quad \overline u_2=\exp(b_3)-\log c_4.\label{eq:hier_ranges}
\end{align}
If $\underline u_2>0$, then the tree is real-admissible on the stated domain.

For AWGN gate evaluations, consider $\widehat U_1=U_1+\xi_1$, $\widehat U_2=U_2+\xi_2$, and $\widehat F_T=G(\widehat U_1,\widehat U_2)+\xi_3$, where \(\xi_q\sim\mathcal N(0,\sigma_q^2/A_q^2)\), \(q\in\{1,2,3\}\), are mutually independent. The two first-level gates have $W_1=2$ and the root layer has $W_2=1$, so Remark~\ref{rem:parallelism_latency} gives $D_1^\star(T)=3$ and $D_P^\star(T)=2$ for $P\ge2$.

The following result gives the end-to-end error and a sufficient positivity condition for the root logarithm.

\begin{proposition}
\label{prop:hierarchical_example}
Assume $|\xi_q|\le\eta_q$ for $q\in\{1,2,3\}$ and $\eta_2<\underline u_2$. Then, we have $\widehat U_2>0$ and
\begin{align}
    |\widehat F_T-F_T(s)|\le \exp({\overline u_1+\eta_1})\eta_1+\frac{\eta_2}{\underline u_2-\eta_2}+\eta_3.\label{eq:hier_error_bound}
\end{align}
\end{proposition}

\begin{IEEEproof}
Since $\widehat U_2\ge\underline u_2-\eta_2>0$, the root logarithm is well defined. The exact and perturbed root inputs lie in
\begin{align}
[\underline u_1-\eta_1,\overline u_1+\eta_1]\times
[\underline u_2-\eta_2,\overline u_2+\eta_2].
\end{align}
Hence, \eqref{eq:local_lipschitz_bound} gives
\begin{align}
    |G(\widehat U_1,\widehat U_2)-G(U_1,U_2)|\le \exp({\overline u_1+\eta_1})\eta_1+\frac{\eta_2}{\underline u_2-\eta_2}.
\end{align}
Combining this bound with $|\xi_3|\le\eta_3$ proves \eqref{eq:hier_error_bound}.
\end{IEEEproof}

Moreover, for $\kappa_q\in(0,1)$, choose
\begin{align}
    \eta_q   = \frac{\sigma_q}{A_q}\Qfun^{-1}\left(\frac{\kappa_q}{2}\right),
    \qquad q\in\{1,2,3\}. \label{eq:hier_local_tolerances}
\end{align}
Then, we have $\Prb[|\xi_q|\leq\eta_q]=1-\kappa_q$. If $\eta_2<\underline u_2$, independence of the three node
errors implies that \eqref{eq:hier_error_bound} holds with probability
$\prod_{q=1}^{3}(1-\kappa_q) \geq 1-\sum_{q=1}^{3}\kappa_q$. On the same bounded-error event, the node-wise peak-power constraints follow from Proposition~\ref{prop:compiler} using the corresponding error-enlarged operand intervals. Such guarantees are therefore conditioned on the bounded-error event. Under coherent flat fading, the same deterministic error bound holds conditionally on both the node-wise channel-inversion feasibility events and the bounded-error events, as used above.


\section{Conclusion}
We derived peak-power-feasible AWGN and coherent flat-fading EML-AirComp gates, resource-latency bounds for prescribed real-admissible EML trees, and a deterministic error recursion resulting in a path-product bound, positivity conditions, and a high-probability AWGN guarantee. A two-hop four-terminal example illustrated the hierarchical analysis, and the digital interface propagates leaf-quantization and local gate errors through the tree. The model assumes that the operands of each gate are available at the corresponding transmitters, and routing, storage, and the existence of real-admissible EML representations for arbitrary functions are left for future work.

\section*{Acknowledgment}
This work was partially supported by German Federal Ministry of Research, Technology and Space (BMFTR) 6GEM+ Transfer Hub under Grants 16KIS2412 and 16KISS005. The author used OpenAI's ChatGPT 5.5 for writing improvements and reviewed and verified all manuscript content.

\bibliographystyle{IEEEtran}
\bibliography{sample}

\end{document}